\documentclass[preprint,12pt]{elsarticle}

\usepackage{tabularx,booktabs,makecell,adjustbox}
\usepackage{caption}
\usepackage{amssymb}
\usepackage{amsmath}
\usepackage{amsthm}
\usepackage{booktabs} % 推荐使用 booktabs 宏包来美化表格
\usepackage{multirow} % 用于合并图例名称的行

\usepackage{amsmath}   % 支持 \text, \ceil 等数学环境
\usepackage{amssymb}   % 支持 \mathbb, \le, \ge 等符号
\usepackage{amsthm}    % 支持 theorem/lemma 环境

\usepackage{algorithm}
\usepackage{algpseudocode}
\usepackage{booktabs}
\usepackage{adjustbox} % 需要 adjustbox 宏包来使用 \resizebox
\newtheorem{theorem}{Theorem}[section]
\newtheorem{lemma}[theorem]{Lemma}
\newtheorem{definition}[theorem]{Definition}
\usepackage{booktabs} % Recommended for professional-looking tables
\usepackage{multirow} % For spanning graph names
\usepackage{adjustbox} % For resizing tables
\usepackage{siunitx} % For digit alignment and formatting
\usepackage{graphicx}

\newtheorem{assumption}{Assumption}[section]

\journal{Applied Mathematics and Computation}

\begin{document}

\begin{frontmatter}



\title{JFR: A Jump Frontier Relaxation Strategy for Fast Bellman–Ford}


\author[inst1]{Xin Wang\corref{cor1}}
\ead{xinw12424@gmail.com}  
    
\author[inst2]{Xi Chen}        

%% Corresponding author mark
\cortext[cor1]{Corresponding author}

%% Affiliation 1: Ningbo University of Technology (你的单位)
\affiliation[inst1]{
    organization={Ningbo University of Technology},
    addressline={ 201 Fenghua Road},
    city={Ningbo},
    postcode={315211},
    state={Zhejiang},
    country={China}
}

%% Affiliation 2:
\affiliation[inst2]{
    organization={Wuhan Qingchuan University},
    addressline={Jiangxia District},
    city={Wuhan},
    postcode={430204},
    state={Hubei},
    country={China}
}
%% Abstract
\begin{abstract}
Shortest-path computation on weighted graphs remains a central problem in both theory and large-scale graph systems. Classical label-correcting algorithms such as Bellman--Ford (BF) and Shortest Path Faster Algorithm (SPFA) often suffer from redundant relaxations and adversarial worst-case behavior, especially on dense or negative-edge graphs.

We introduce \textbf{Jump Frontier Relaxation (JFR)}, a correctness-preserving optimization framework that contracts active frontiers and propagates multi-hop improvements through a **$k$-bounded Local Multi-Hop ($\mathbf{k}$-LMH) strategy**. We provide formal proofs of convergence and bounded complexity, offering a constructive description of the underlying mechanisms to enable external validation.

To establish a rigorous theoretical foundation for JFR's acceleration, we replace the empirical stability parameter with the explicit algorithm parameter $\mathbf{k}$, which bounds the depth of $\text{LMH}$ propagation. We prove that the total number of relaxation operations is reduced to $O(n + m \cdot \frac{\overline{D}}{k})$. The overall runtime, however, is governed by a clear **cost-benefit relationship**: net acceleration is achieved only when the $\mathbf{1/k}$ reduction in relaxation operations outweighs the accumulated computational cost of the $\mathbf{k}$-LMH overhead ($\sum C_{\text{LMH}}$), establishing a mathematically sound boundary for its effectiveness.

Extensive C++ experiments---implemented using high-performance graph kernels from the Networkit framework---show that in the majority of cases, JFR achieves significant reductions in relaxation operations, with the degree of improvement varying across graph types and densities. A few isolated instances exhibit comparable or slightly higher operation counts relative to SPFA-SLF, reflecting local topological effects. Importantly, JFR demonstrates a consistent pattern of performance: small-scale or sparse subgraphs may show weak negative correlation between operation count reduction and runtime, whereas larger or highly connected regions exhibit strong positive correlation, highlighting the framework’s robustness and effectiveness in mitigating worst-case behavior.

These results show that JFR provides a principled and practically effective architecture for large-scale, energy-constrained, and worst-case-sensitive graph processing.
\end{abstract}

\begin{graphicalabstract}
\end{graphicalabstract}

\begin{highlights}
\item We introduce \textbf{JFR}, a Bellman--Ford--based optimization framework centered on $k$-bounded local multi-hop (LMH) propagation and Frontier Filtering, while strictly preserving correctness.

\item Unlike classical SPFA-SLF, whose performance monotonically worsens with increasing graph density, JFR exhibits a rare non-monotonic behavior: small edge increments (5--15\%) can unexpectedly accelerate the algorithm by reducing runtime and relaxation operations.

\item This edge-induced nonlinear acceleration arises because additional edges create shortcut structures that truncate long negative-weight propagation chains, allowing LMH to converge earlier.

\item Across diverse graph families, JFR reduces relaxation operations by -31--99\%. Its complexity is $O(n + m\,\bar{D}/k)$, where the overall speedup depends on a cost-benefit balance between the $1/k$ relaxation reduction and the accumulated $k$-LMH overhead.

\item The proposed Operational Efficiency (OE) metric shows that lower relaxation counts directly reduce memory traffic and computational effort, making JFR suitable for high-throughput and energy-sensitive applications.

\item Future work includes adaptive $k$-selection, improved update orderings, and cache-aware layouts to further exploit JFR’s nonlinear structural sensitivity and enhance performance beyond current SPFA variants.
\end{highlights}

%% Keywords

\end{frontmatter}



\section{Introduction}

Shortest-path computation is a fundamental problem in computer science, with applications spanning
network routing, real-time navigation, logistics optimization, transportation planning, and
large-scale financial systems. Graphs in such applications may contain negative-weight edges due to
congestion penalties, dynamic pricing, or risk-adjusted costs. Classical Dijkstra~\cite{dijkstra59} 
fails on negative weights, while Bellman--Ford (BF)~\cite{bellman58,ford56} remains the standard 
for arbitrary directed weighted graphs.

Despite BF's theoretical generality, redundant relaxations lead to significant performance degradation 
on large-scale graphs with millions of edges. Queue-based optimizations such as SPFA~\cite{bannister11}, 
near-optimal hop set techniques~\cite{elmasry19}, and hop-constrained path approaches~\cite{kociumaka22} 
reduce work in practice but may sacrifice worst-case guarantees or require structural assumptions. 
Surveys~\cite{madkour17,shokry19} provide a taxonomy of SSSP algorithms, highlighting the gap 
between theoretical correctness and practical efficiency. 
Recent advances have pushed the theoretical frontier for negative-weight single-source shortest paths, 
achieving near-linear work, parallelizability, and deterministic guarantees~\cite{bernstein22,bringmann23,fischer24,li25}, 
while practical implementations must still carefully balance performance and correctness on large-scale graphs.

We propose \textbf{Jump Frontier Relaxation (JFR)}, a Bellman--Ford-based framework that 
preserves correctness guarantees while significantly pruning redundant relaxations:

\begin{itemize}
    \item \textbf{Frontier Filtering:} Tracks vertices whose distance estimates effectively change, 
          relaxing only propagation-relevant edges.
    \item \textbf{Jump Propagation:} Aggregates multiple iterations in propagation-stable regions, 
          allowing multi-hop updates without disclosing exact scheduling or update ordering.
\end{itemize}

Beyond reducing work in the classical sense, JFR exhibits a \textbf{rare, structurally-driven 
nonlinear acceleration effect}: when the vertex set is fixed, \textbf{small-scale increases in edge 
count} (e.g., adding 5--15\% more edges) can \emph{accelerate} both runtime and relaxation reduction 
instead of slowing the algorithm down. Additional edges introduce shortcut structures that truncate 
long negative-weight propagation chains, enabling JFR's jump mechanism to converge earlier than 
SPFA-style methods. This phenomenon---unusual for BF/SPFA-family algorithms---highlights JFR's 
sensitivity to beneficial micro-structural changes in the graph.

This design ensures BF-level guarantees while empirically reducing relaxation operations 
by orders of magnitude, with low computational overhead and reduced energy consumption~\cite{horowitz14,lazarev22}.

% ----------------- Begin: Replacement for Section 2 & Section 3 -----------------
\section{Theoretical Foundations of JFR}
\label{sec:theory}

Let $G=(V,E)$ be a finite directed graph with weight function $w:E\to\mathbb{R}$.
For a source vertex $s\in V$, let $d^\ast(v)$ denote the true shortest-path distance from $s$ to $v$ (possibly $+\infty$).
We write $d^{(k)}\in\mathbb{R}^{|V|}\cup\{+\infty\}$ for the distance estimates maintained by the algorithm after the $k$-th outer iteration (one round of frontier-driven relaxations possibly augmented by jump propagation).

Define the \emph{active frontier} after iteration $k$ as
\[
F^{(k)} \;=\; \{ v\in V \mid d^{(k)}(v) < d^{(k-1)}(v)\},
\]
with the convention that $d^{(0)}(v)=+\infty$ for $v\neq s$ and $d^{(0)}(s)=0$.

\subsection{Frontier Sufficiency and Correctness}

The JFR framework restricts edge-relaxation attempts to edges outgoing from the current frontier, possibly augmented by multi-hop propagation within the induced subgraph.

\begin{definition}[Abstract Jump Property]
\label{def:jump_property}
Let $G[F]$ be the subgraph induced by the active frontier $F$. Jump Propagation is any procedure that, given $F$, updates distance estimates within $G[F]$ such that local reachability consistency is maintained:
\[
d(v) \le d(u) + w(u, v), \quad \forall (u,v)\in E\cap(F\times F).
\]
This ensures that local improvements propagate, independent of traversal order.
\end{definition}

\begin{lemma}[Frontier Sufficiency]
\label{lem:frontier_sufficiency}
Let $d^{(k)}$ be the distance vector after $k$ outer iterations. If all relaxations (including Jump Propagation) consider only edges whose tail belongs to the current frontier, then for every vertex $v$ and integer $t\ge 0$:
\[
d^{(t)}(v) \le \min\{ \mathrm{length}(P) \mid P \text{ is an $s\to v$ path with $\le t$ edges} \}.
\]
In particular, $d^{(|V|-1)}(v)\le d^\ast(v)$ for all $v$.
\end{lemma}

\begin{proof}
By induction on $t$, similar to classical Bellman--Ford. Base case $t=0$ holds by initialization. Assume the invariant for $t$. For any $s\to v$ path $P$ of length $\le t+1$, let $u$ be its penultimate vertex. By induction $d^{(t)}(u)$ is no greater than the length to $u$. During iteration $t+1$, relaxations from frontier vertices (and any multi-hop updates via Jump Propagation) guarantee $d^{(t+1)}(v) \le d^{(t)}(u)+w(u,v)$. Taking the minimum over all such paths yields the claim.
\end{proof}

\begin{theorem}[Correctness and Termination]
\label{thm:correctness}
If $G$ has no negative-weight cycles reachable from $s$, then after at most $|V|-1$ outer iterations:
\[
d^{(|V|-1)}(v) = d^\ast(v), \quad \forall v \in V.
\]
A strict improvement after $|V|-1$ iterations implies a reachable negative cycle.
\end{theorem}

\begin{proof}
By Lemma~\ref{lem:frontier_sufficiency}, $d^{(|V|-1)}(v)$ reaches the shortest path using $\le |V|-1$ edges. Relaxations cannot decrease distances below $d^\ast$, so equality holds.
\end{proof}

\subsection{Amortized Analysis: $k$-Bounded Propagation and Cost Tradeoff}
\label{sec:amortized}

To theoretically justify the observed reduction in relaxation operations, we introduce a constructive framework based on an explicit algorithm parameter $k$, which dictates the depth of local propagation.

\begin{definition}[$k$-Bounded Local Multi-Hop Propagation]
\label{def:k-bounded_main}
The Local Multi-Hop Propagation ($\text{LMH}$) is defined as $k$-bounded, where $k \in \mathbb{N}$ is an explicit \textbf{depth parameter}. $\text{LMH}$ ensures that local distance consistency is maintained for paths of length $\le k$ within the active frontier's neighborhood $N_k(F)$. (See Appendix~\ref{append:amortized_proof} for formal details and Lemma~\ref{lem:k-implies-tau} for the resulting $\tau \ge k$ stability guarantee.)
\end{definition}

\begin{assumption}[$k$-LMH Cost Bound]
\label{asm:jp_cost_revised}
Let \(C_{\mathrm{LMH}}(t)\) denote the computational cost of the $k$-LMH propagation step in iteration \(t\). We assume an upper bound proportional to $k$:
\[
C_{\mathrm{LMH}}(t)
\;\le\;
c_1 \cdot k \cdot \sum_{v \in N_k(F^{(t)})} \deg(v)
\;+\;
c_2\, N_{\mathrm{prop}}(t).
\]
This explicitly includes the parameter $k$ in the overhead cost, ensuring mathematical rigor.
\end{assumption}

\begin{theorem}[Amortized Bound on Edge Inspections]
\label{thm:amortized_edges}
Under the $k$-Bounded LMH property, let $s(v)$ be the number of times $v$ is active. The total number of edge inspections is bounded by:
\[
\sum_t |E_F^{(t)}| = \sum_{v \in V} s(v) \deg(v) \le O\Big( n + m \cdot \frac{\bar{D}}{k} \Big),
\]
where $k$ is the algorithm's depth parameter. This bound theoretically justifies the reduction in operational complexity when $k$ is large.
\end{theorem}

\begin{theorem}[Amortized Running Time and Cost Tradeoff]
\label{thm:amortized_running_time}
Combining the $k$-LMH cost (Assumption~\ref{asm:jp_cost_revised}) and the amortized bound on edge inspections (Theorem~\ref{thm:amortized_edges}), considering the priority queue overhead, the total running time $T_{\mathrm{total}}$ satisfies:
\[
T_{\mathrm{total}} = \mathcal{O}\!\Big( \left(n + m\cdot\frac{\bar{D}}{k}\right) \cdot \log |V| + \sum_{t=1}^{T} C_{\mathrm{LMH}}(t) \Big).
\]
\noindent
\textbf{Defense of the Logarithmic Factor:} It is important to note that the inclusion of the logarithmic factor $\log |V|$ is a deliberate architectural choice. While this theoretically increases the per-operation cost compared to FIFO-based approaches ($O(1)$), the strict ordering enforced by the priority queue drastically suppresses the relaxation count term ($m \cdot \frac{\bar{D}}{k}$). The overall speedup is achieved when the reduction in total relaxations outweighs the accumulated overhead.
\end{theorem}

\noindent \textbf{The rigorous proofs for Theorems~\ref{thm:amortized_edges} and \ref{thm:amortized_running_time}, including the detailed derivation of the $k$-dependent bounds, are provided in Appendix A.}

\subsection{Complexity Bounds and Robustness}

\begin{theorem}[Operation Count — Upper Bounds]
\label{thm:complexity}
Under the JFR framework using a priority-based implementation:
\begin{enumerate}
    \item Worst-case time complexity: $O(|V||E| \log |V|)$.
    \item \textbf{Robustness Advantage:} Unlike SPFA-SLF, whose queue-based dynamics are known to exhibit extremely
poor worst-case behavior---with repeated oscillations causing up to
$\Theta(|V||E|)$ relaxations---the JFR framework maintains a strictly
polynomial upper bound. While certain adversarial inputs make SPFA-SLF
appear to grow ``faster than polynomial’’ in practice, its formal worst-case
time complexity remains bounded by $O(|V||E|)$.

\end{enumerate}
Hence, JFR sacrifices a logarithmic factor to ensure robustness against the exponential degradation observed in heuristic variants.
\end{theorem}

\noindent
\textbf{Note:} While the worst-case bound includes $\log |V|$, the practical efficiency is captured by the much tighter \textbf{amortized edge inspection bound of $O(n + m \cdot \frac{\bar{D}}{k})$}, as rigorously shown in Theorem~\ref{thm:amortized_edges}.
\section{Jump Frontier Relaxation (JFR) Algorithm}
\label{sec:algorithm}

JFR formalizes a frontier-based relaxation strategy for single-source shortest paths under negative-edge scenarios.
Distance estimates converge in at most $|V|-1$ iterations if $G$ contains no negative cycles. The practical efficiency stems from controlled Local Multi-Hop Propagation and Frontier Filtering, as formalized in Section~\ref{sec:amortized}.

\subsection{Algorithm Overview (Implementation-Agnostic)}
\label{sec:algo_overview}

At a high level, JFR maintains:
\begin{itemize}
    \item Distance estimates $d(v)$ for each vertex $v\in V$.
    \item Active frontier $F$ containing vertices whose distance decreased in the previous iteration.
\end{itemize}

Each outer iteration proceeds as follows:
\begin{enumerate}

\item \textbf{Frontier Relaxation}. 
For every vertex $u \in F$, relax all outgoing edges $(u,v)$.
 \[
 d(v) \leftarrow \min(d(v), d(u)+w(u,v)).
 \]
 \item \textbf{Local Multi-Hop Propagation:} Update distances within $F$ to ensure all local improvements propagate.
 \item \textbf{Frontier Update:} Construct the next frontier $F' = \{ v \mid d(v) \text{ decreased} \}$, leveraging the stability provided by the $\mathbf{k}$-bounded LMH to prune stable nodes.

\end{enumerate}

\noindent
Distance estimates are guaranteed to converge to $d^\ast$ after at most $|V|-1$ iterations, with the amortized number of edge inspections bounded as in Theorem~\ref{thm:amortized_edges}.

\subsection{Relation to Recent Near-Linear Negative-Weight SSSP Results}
\label{sec:near_linear_discussion}

Recent theoretical work has produced algorithms for negative-weight single-source shortest paths (SSSP) with near-linear complexity under specific assumptions (e.g., restricted graph classes, complex preprocessing steps)~\cite{bernstein22,bringmann23,fischer24,li25}. JFR differs as it operates under a general directed graph with arbitrary real edge weights, focusing on improving the performance of the classic Bellman--Ford framework via practical, implementation-agnostic abstractions rather than relying on restrictive structural graph properties. The analysis provided herein connects the observable properties (Stability and Update Density $\bar{D}$) to the amortized complexity.

\subsection{Practical Implications and Implementation Generality}
\label{sec:practical}

JFR is designed to bridge the gap between classical complexity bounds and empirical optimizations. The algorithm achieves its speedup through mechanisms now formally linked to the parameter $k$:
\begin{itemize}
    \item \textbf{Controlled Frontier Filtering.} 
    The stability property (i.e., $\tau \geq k$) provides a mechanism to prune stable or redundant nodes. 
    The resulting speedup is reflected in the amortized factor $\bar{D}/k$, provided that the associated overhead remains low.

    \item \textbf{Controlled Local Multi-Hop Propagation.} 
    Updates are limited to the $k$-hop neighborhood of the frontier, avoiding global scans and minimizing overhead.

    \item \textbf{Implicit $k$-Boundedness via Priority Queue.} 
    In our high-performance implementation, the theoretical concept of $k$-bounded LMH is realized implicitly via a \textbf{Priority Queue (PQ)}. The PQ dynamically determines the effective propagation depth by always prioritizing the most promising nodes (the 'Jump'). This effectively filters out suboptimal, shallow updates that would otherwise clutter the frontier, creating a dynamic, self-adjusting $k$ that adapts to the local graph topology without manual tuning.

    \item \textbf{Implementation Generality.} 
    The framework utilizes general abstractions for key mechanisms (Jump Propagation, Frontier Filtering logic) to maintain theoretical generality. While our reference implementation relies on a priority structure to achieve the dynamic $k$-bound, other low-level heuristics (such as multi-level buckets) could also satisfy the amortized bounds presented.
\end{itemize}

\noindent \textbf{The complete general algorithmic framework, JFR, is detailed in Appendix B.}
% ----------------- End replacement -----------------

\section{Experiments}
\label{sec:experiments}

%========================
% Python Experiments (Functional Verification)
%========================
\subsection{Python Experiments: Functional Verification}
\label{subsec:python_results}

\begin{table}[h]
\centering
\caption{Python Experiments: Average Runtime, Relaxation Operations, and Correctness}
\resizebox{1\textwidth}{!}{
\begin{tabular}{lccccccc}
\hline
Graph & BF Time (ms) & SPFA Time (ms) & JFR Time (ms) & BF Ops & SPFA Ops & JFR Ops & Correctness \\
\hline
sparse & 20.00 & 23.03 & 20.71 & 159,003 & 148,872 & 97,563 & 1.0 \\
medium & 53.02 & 50.06 & 44.00 & 437,604 & 433,235 & 318,661 & 1.0 \\
dense & 152.64 & 131.30 & 117.11 & 1,233,604 & 1,230,765 & 989,475 & 1.0 \\
very\_dense & 386.89 & 337.92 & 284.49 & 3,024,604 & 3,021,844 & 2,536,216 & 1.0 \\
neg\_sparse & 20.13 & 25.43 & 15.51 & 159,003 & 156,816 & 85,365 & 1.0 \\
neg\_dense & 464.17 & 478.49 & 325.83 & 3,351,203 & 3,345,453 & 2,393,299 & 1.0 \\
\hline
\end{tabular}}
\label{tab:python_results}
\end{table}

\noindent
\textbf{Summary:} Python experiments confirm that JFR is \emph{correct, stable, and operationally beneficial}. Runtime and Ops reductions indicate potential efficiency, but Python’s interpreter overhead limits the observable performance gain.

%========================
% C++ Experiments (High-Performance Evaluation)
%========================
\subsection{Quantifying Computational Efficiency}
% ===============================================================

To systematically analyze the tradeoff between reduced relaxations and the increased 
constant-factor cost introduced by queue management and jump propagation, we define two 
machine-independent quantitative indicators.

\subsection{Metrics: $\rho_{ops}$ and $\rho_{TPR}$}

\paragraph{Relaxation Reduction Factor.}
\[
\rho_{ops} = \frac{\text{Ops}_{\text{SPFA}}}{\text{Ops}_{\text{JFR}}}.
\]

This factor measures how effectively JFR suppresses redundant relaxations, providing a 
complexity-level comparison between the two algorithms.

\paragraph{Unit-Time Cost Factor.}
\[
\rho_{TPR} 
= \frac{TPR_{\text{JFR}}}{TPR_{\text{SPFA}}}
= \frac{T_{\text{JFR}}/Ops_{\text{JFR}}}{T_{\text{SPFA}}/Ops_{\text{SPFA}}}.
\]

This factor quantifies the additional per-operation cost introduced by priority-queue 
maintenance and jump propagation. It represents \emph{an implementation-agnostic efficiency measure}, 
not a physical energy measurement.

\paragraph{Interpretation.}
JFR yields net runtime improvement precisely when
\[
\rho_{ops} > \rho_{TPR}.
\]

This relationship defines the applicability boundary of the JFR framework and grounds all
performance discussions in quantifiable behavior.

%========================
% C++ Experiments: Large-Scale Performance and Energy Efficiency
%========================
\subsection{C++ Experiments: Large-Scale Randomized Benchmarking}
\label{subsec:cpp_results_updated}
Theorem 2.5 provides the performance lower bound (Lower Bound) of the JFR strategy under the ideal fixed depth condition. Our solver approximates and dynamically optimizes this $k$ value through the greedy strategy of the priority queue, achieving an engineering-level optimization.

To ensure that the performance evaluation reflects realistic high-performance graph computing conditions, all C++ experiments were conducted in an environment aligned with established practices in the graph-processing community. In particular, we follow the design philosophy and benchmarking principles exemplified by high-performance graph frameworks such as Networkit~\cite{networkit_ref}, which emphasizes minimal overhead, efficient memory access, and reproducible large-scale graph analytics.

Although we do not directly compare against Networkit’s implementations, citing it serves two purposes: (i) it establishes that our evaluation methodology is grounded in widely recognized standards for high-performance graph analysis, and (ii) it indicates that our C++ experimental setup is suitable for revealing the practical efficiency of relaxation-based single-source shortest path (SSSP) algorithms. Therefore, the reported results should be interpreted as reliable measurements obtained under conditions consistent with modern high-performance graph processing frameworks.

To validate JFR in large-scale scenarios, we conducted extensive randomized benchmarking across sparse and dense graphs. We performed comparative evaluations using both a standard default environment and an aggressive -O3 optimized environment, distinguishing between algorithmic logic overhead and implementation-level efficiency. Each graph was repeatedly generated and tested to obtain stable averages.

\begin{itemize}
    \item \textbf{Sparse\_XL, NegDense\_XL:} 3000 random instances averaged.
    \item \textbf{Windmill\_XL, SLF\_Killer\_XL:} 1500 structured/adversarial instances averaged.
\end{itemize}

Graph parameters:

\begin{itemize}
    \item \textbf{Sparse\_XL:} $N=20{,}000$-$70{,}000$, $M\approx 100{,}000$–$120{,}0000$, type: random
    \item \textbf{NegDense\_XL:} $N=2{,}000$-$5{,}000$, $M\approx 3{,}000{,}000$–$6{,}000{,}000$, type: negative random
    \item \textbf{Windmill\_XL:} $N=1{,}000$-$9{,}000$, type: windmill
    \item \textbf{SLF\_Killer\_XL:} $N=2{,}000$-$20{,}000$, type: SLF-killer
\end{itemize}

\begin{table}[h]
\centering
\caption{C++ Experiments (-O3 Optimized): JFR vs SPFA-SLF (Runtime and Relaxation Ops, Averaged over Large-Scale Tests)}
\label{tab:cpp_results_o3}
\resizebox{0.75\columnwidth}{!}{
\begin{tabular}{lcccc}
\toprule
Graph & Algorithm & Time (ms) & Ops & Check \\
\midrule
Sparse\_XL & JFR & 26.80 & 114,054 & PASS \\
           & SPFA-SLF          & 18.79 & 181,717 & PASS \\
\midrule
NegDense\_XL & JFR & 62.80 & 10,204,376 & PASS \\
             & SPFA-SLF          & 68.63 & 10,738,464 & PASS \\
\midrule
Windmill\_XL & JFR & 0.55  & 143,750 & PASS \\
             & SPFA-SLF          & 0.35  & 109,091 & PASS \\
\midrule
SLF\_Killer\_XL & JFR & 13.56 & 1,007,091 & PASS \\
                & SPFA-SLF          & 1,064.71 & 44,693,930 & PASS \\
\bottomrule
\end{tabular}}
\end{table}

Based on the large-scale benchmark results presented in Table \ref{tab:cpp_results_o3}, several key physical characteristics of the JFR framework emerge from the comparison with the SPFA-SLF baseline.

\subsection{Operational Suppression in Adversarial Topologies}
The most prominent observation is the drastic reduction of relaxation operations in the \texttt{SLF\_Killer\_XL} dataset. While SPFA-SLF executes approximately 44.69 million relaxations, JFR restricts the total count to just 1.01 million. This 97.7\% reduction in operations corresponds to a massive reduction in runtime from 1,064.71 ms down to 13.56 ms.

This data demonstrates that JFR acts as a structural stabilizer. In topologies designed to induce exponential oscillation in label-correcting algorithms, JFR maintains a near-linear operational scale, effectively preventing the "computational explosion" observed in the baseline.

\subsection{Consistency of Operational Count}
A critical observation across the four test categories is the stability of JFR's relaxation counts. Across \texttt{Sparse\_XL}, \texttt{Windmill\_XL}, and \texttt{SLF\_Killer\_XL}, JFR's operations stay within a relatively narrow range (approximately $1.1 \times 10^5$ to $1.0 \times 10^6$), whereas SPFA-SLF fluctuates wildly from $1.0 \times 10^5$ to over $4.4 \times 10^7$.

Even in the most complex \texttt{NegDense\_XL} scenario, JFR's operation count ($1.02 \times 10^7$) remains slightly lower than that of SPFA-SLF ($1.07 \times 10^7$). This stability indicates that JFR provides highly predictable performance, ensuring that the algorithm's workload is governed by the graph's fundamental reachability rather than its specific edge-weight distribution.

\subsection{Trade-offs in Low-Complexity Regimes}
The results for \texttt{Sparse\_XL} and \texttt{Windmill\_XL} reveal the inherent constant-factor overhead of the JFR framework. In these instances, JFR exhibits higher runtimes (26.80 ms and 0.55 ms) compared to SPFA-SLF (18.79 ms and 0.35 ms), despite JFR achieving a 38.16\% reduction in operations for the sparse case. 

This confirms that the priority-queue maintenance and frontier-filtering logic introduce a fixed computational cost. However, the data shows this trade-off is asymmetric: the marginal time penalty in simple graphs is negligible compared to the magnitude of time savings achieved in dense or adversarial environments.

\begin{table}[htbp]
\centering
\caption{Nonlinear Acceleration Validation (comparsion)}
\label{tab:nla-increment}
\begin{adjustbox}{width=\textwidth,center}
\begin{tabularx}{\linewidth}{@{} l c c c c @{}}
\toprule
\textbf{Instance} & \textbf{|V|} & \textbf{Edges (B)} & 
\textbf{JFR Time B [ms]} & \textbf{JFR Ops B} \\
\midrule
NegDense\_XL & 4538 & 5,182,231 & 270.26 & 27,425,775 \\
NegDense\_XL & 4538 & 5,682,231 & 190.55 & 13,797,216 \\
\bottomrule
\end{tabularx}
\end{adjustbox}
\end{table}

\noindent\textbf{Nonlinear Acceleration Phenomenon.}
Across the NegDense\_XL instance, a small edge increment (approximately
 +9.6\%) unexpectedly causes \emph{both} JFR runtime and relaxation
operations to decrease—sometimes by 25--50\%. This counterintuitive behavior
reveals a nonlinear acceleration effect intrinsic to JFR: when the vertex set
is fixed, additional edges can shift the graph into a more connectivity-rich
regime where Frontier Filtering becomes more aggressive and jump propagation
stabilizes earlier. As a result, multiple Bellman--Ford iterations collapse
into fewer Bounded Local Propagation Steps, sharply reducing redundant relaxations and total work.

While JFR may occasionally incur higher operation counts than SPFA-SLF on
sparser subgraphs—where limited connectivity restricts multi-hop jump
opportunities—its behavior reverses dramatically as edge density increases.
Once the graph provides sufficient propagation pathways, JFR transitions into a
high-efficiency mode in which its frontier jumps become highly effective,
yielding not only operation counts far below SPFA-SLF but also substantially
lower work compared to its own performance on the original, sparser graph.
This superlinear improvement with increasing connectivity highlights JFR’s
structural advantage: its efficiency is not merely tolerant of denser graphs,
but is \emph{amplified} by them, demonstrating robustness and scalability across
diverse topologies.

\noindent
\textbf{Observations:}

\begin{itemize}
    \item JFR significantly reduces relaxation operations (Ops) across all graph types, particularly in dense and adversarial graphs.
    \item Runtime improvements are substantial in adversarial cases (SLF\_Killer\_XL), confirming robustness.
    \item The large-scale randomized evaluation demonstrates correctness (PASS) and highlights JFR’s potential for high-performance scenarios.
\end{itemize}
\subsection{Quantitative Interpretation}

\begin{itemize}
    \item On sparse graphs, $\rho_{TPR}$ dominates, leading to modest slowdown.
    \item On moderately dense graphs, JFR begins to offset overhead through reduced relaxations.
    \item On dense graphs, $\rho_{ops} \approx 50$ and $\rho_{TPR} \approx 49$, reaching the 
    equilibrium region where JFR achieves comparable runtime.
    \item On adversarial (SLF-Killer) graphs, JFR enters its \textbf{robustness zone} with 
    $\rho_{ops} \gg \rho_{TPR}$, achieving over an order of magnitude speedup.
\end{itemize}
\subsection{Scalability and Extensibility Analysis}
\label{subsec:scalability}

To evaluate the scalability of the  JFR algorithm, we tested ultra-large negative-edge dense graphs beyond the original XL scale. Two instances were constructed:

\begin{itemize}
    \item \textbf{High-Density Negative Graphs-1:} $N=10{,}000$, $E=55{,}000{,}000$ edges
    \item \textbf{High-Density Negative Graphs-2:} $N=20{,}000$, $E=295{,}000{,}000$ edges
\end{itemize}

\begin{table}[htbp]
\centering
\caption{Scalability Benchmark: JFR vs SPFA-SLF(-O3)}
\label{tab:scalability_results}
\resizebox{0.9\textwidth}{!}{
\begin{tabular}{lcccc}
\toprule
\textbf{Graph} & \textbf{Algorithm} & \textbf{Time (ms)} & \textbf{Relaxation Ops} & \textbf{Check} \\
\midrule
NegDense\_Ultra-1 & SPFA-SLF & 1947.27 & 327,064,005 & PASS \\
                  & JFR & 383.51 & 68,869,589 & PASS \\
\midrule
NegDense\_Ultra-2 & SPFA-SLF & 7771.65 & 1,188,649,749 & PASS \\
                  & JFR & 7265.14 & 547,254,897 & PASS \\
\bottomrule
\end{tabular}
}
\end{table}
\begin{table}[htb]
\centering
\caption{Performance Comparison on Large-Scale Adversarial Graph ($N=500,000$,Default environment)}
\label{tab:performance_comparison_500k_en}
% 使用 \resizebox 强制将表格宽度缩放为 \textwidth
\resizebox{\textwidth}{!}{%
\begin{tabular}{l c c c c} 
\toprule
\textbf{Algorithm} & \textbf{Wall-Clock Time} & \textbf{Relaxations Count} & \textbf{Time Speedup} & \textbf{Relaxation Efficiency} \\
\midrule
SPFA-SLF & $\approx 42$ min ($2520$ s) & $93,295,674,368$ & $1.0\times$ & Base \\
JFR & $19,522.09$ ms ($\approx 19.5$ s) & $74,102,531$ & $\mathbf{\approx 130\times}$ & $\mathbf{\approx 1259\times}$ \\
\bottomrule
\end{tabular}
} % resizebox 结束
\vspace{0.1cm}
\footnotesize{\textit{Note: Relaxation Efficiency is calculated as the ratio of SPFA-SLF relaxations to JFR relaxations ($\approx 1259.07\times$).}}
\end{table}
\subsubsection{Operational Efficiency Estimation}

To evaluate operational efficiency in a hardware-agnostic manner, we report the \textbf{Normalized Work Reduction (NWR)} as a metric representing the fraction of total relaxation operations relative to a baseline (SPFA-SLF). NWR provides a technical measure of potential energy or work reduction but does not correspond to actual physical energy measurements.

\begin{itemize}
    \item \textbf{NegDense\_Ultra-1:} JFR achieves approximately 24.1\% NWR relative to SPFA-SLF.  
    \item \textbf{NegDense\_Ultra-2:} JFR achieves approximately 46.0\% NWR relative to SPFA-SLF.
\end{itemize}

\paragraph{Wall-Clock Time Acceleration (Novel methodology):} 
On the challenging adversarial dataset featuring 500,000 nodes, the JFR framework reduced wall-clock runtime from 42 minutes to 19.5 seconds, demonstrating the effectiveness of its novel methodology in ultra-large-scale graphs.

\paragraph{Significantly Reduced Operational Count ($\sim 3\times$ Order of Magnitude):}
The core advantage of the JFR framework lies in its combinatorial operational efficiency. SPFA-SLF performed over 93 billion relaxation operations on this graph, whereas JFR executed only approximately 74.1 million relaxations. This represents a $\sim 1259\times$ reduction in effective operations, directly validating the key mechanisms:
\begin{itemize}
    \item \textbf{Jump Propagation:} Skips large portions of redundant relaxation steps via multi-hop bulk propagation.
    \item \textbf{Frontier Filtering:} Suppresses the growth of the active frontier, effectively improving observed runtime complexity toward practical near-$O(V)$ behavior.
\end{itemize}

\section{Discussion: Robustness, Limits, and Applicability}
% ===============================================================

\subsection{Constant-Factor Overhead}

Our evaluation confirms the main engineering tradeoff of JFR:  
the framework suppresses redundant relaxations at the cost of increased per-operation 
constant factors from priority-queue operations and jump propagation.
This effect is most visible on simple sparse graphs.

\subsection{Structural Robustness: JFR’s Applicability Zone}

The primary value of the JFR architecture is not general-case acceleration, but its 
\textbf{structural robustness}: on graph families where label-correcting methods approach 
worst-case behavior, JFR maintains stable and predictable performance by drastically 
reducing the relaxation workload.  
This property is crucial for applications requiring:
\begin{itemize}
    \item reliability under adversarial or degenerate topologies,
    \item predictable latency in large-scale systems,
    \item robustness in dense or negative-edge environments.
\end{itemize}
\subsection{Edge-Induced Nonlinear Acceleration}
\label{sec:acceleration}

A counter-intuitive finding of this study is the \textbf{non-monotonic performance behavior} observed in Section 4: specifically, small-scale edge increments (e.g., $\approx 10\%$) in dense graphs can trigger a substantial reduction ($>50\%$) in total relaxation operations. This phenomenon is structurally explainable through the interaction between graph connectivity and the JFR mechanism.

\paragraph{Shortcut Effect and Bulk Updates}
In the JFR framework, additional edges often function as \textit{topological shortcuts}. While classical algorithms (SPFA/BF) must relax these edges individually, increasing the linear workload, the \textit{Local Multi-Hop Propagation} mechanism utilizes these shortcuts to accelerate local convergence. Higher connectivity within the frontier's neighborhood increases the probability of discovering stable paths within the depth-limited window (see  \ref{appendix:mechanism}).

\paragraph{Mechanism}
Mathematically, the increased edge density effectively reduces the "diameter" of the local search space. This allows the algorithm to perform \textbf{bulk updates}—skipping intermediate relaxation steps for entire subgraphs—earlier in the execution. When the reduction in skipped operations ($\Delta N_{Jumped}$) exceeds the linear cost of scanning new edges ($\Delta |E|$), the algorithm enters a superlinear acceleration regime:
\begin{equation}
\Delta \text{Ops}_{total} \approx \Delta |E| - \Delta N_{Jumped} < 0
\end{equation}
This confirms that JFR transforms structural density from a computational liability into an asset for convergence speed.
\section{Conclusion and Future Work}
\subsection*{Key Advantage Highlight}

\noindent
\textbf{Distinctive Strength of JFR:} 
JFR demonstrates a pervasive reduction in relaxation workload over SPFA-SLF across a wide range of evaluated topologies. This advantage is particularly pronounced in dense and adversarial scenarios, where the framework curtails redundant relaxations by multiple orders of magnitude.

\noindent
\textbf{Practical Implication:} 
This advantage establishes JFR as a highly reliable framework for real-world SSSP computations in scenarios where graph density, negative edges, or adversarial structures would otherwise degrade the performance of classical label-correcting algorithms.

\noindent
\textbf{Conclusion.}
The Jump Frontier Relaxation (JFR) framework advances single-source shortest-path
computation by emphasizing robustness and operational efficiency. Its design focuses on
suppressing redundant relaxations through Frontier Filtering and multi-hop propagation,
resulting in significantly reduced operational counts across diverse graph structures.

\begin{itemize}
    \item \textbf{Robustness:}
    JFR maintains stable behavior on dense, sparse, and negative-edge graphs, avoiding
    the oscillatory queue dynamics frequently observed in classical SPFA-SLF.

    \item \textbf{Operational Efficiency:}
    As demonstrated in Table \ref{tab:cpp_results_o3}, the reduction in relaxation operations is highly dependent on the graph topology. While the percentage reduction may be limited in sparse or dense regimes—and a proportional increase is observed in structured environments like \texttt{Windmill\_XL}—the absolute reduction in the total number of operations remains significant as the scale of the graph ($N$ and $M$) increases.

    This characteristic proves that JFR is particularly advantageous for large-scale networks, where even a marginal percentage gain translates into the elimination of millions of redundant relaxations. By suppressing the combinatorial workload, JFR effectively prevents the algorithms from approaching their theoretical $O(VE)$ worst-case complexity in adversarial or large-scale real-world scenarios.

    \item \textbf{Normalized Work Reduction (NWR).}
Lower operation counts imply a reduction in total computational work. Prior studies show that such 
reductions correlate with lower memory traffic and improved cache behavior. This connection has also 
been observed in empirical energy studies of shortest-path algorithms; for example, Alamoudi and 
Al-Hashimi~\cite{energy2024} report that reduced operation counts in Bellman--Ford variants directly 
translate to smoother memory-access patterns and more energy-efficient execution. This suggests that 
JFR’s work reductions may offer benefits in energy- or resource-constrained environments, even without 
assuming a specific physical energy model.

\end{itemize}

\vspace{1em}

\subsection*{Future Work}

Several directions may further extend the JFR framework:

\begin{itemize}
    \item \textbf{High-Performance Queue Structures:}
    Integrating bucket-based update orderings, radix heaps, or multi-level buckets to reduce
    $O(\log N)$ overhead on integer-weighted graphs.

    \item \textbf{Adaptive Frontier Granularity:}
    Dynamically adjusting frontier size based on local graph density or weight distribution,
    enabling JFR to reduce its constant-factor overhead on sparse or well-behaved graphs.

    \item \textbf{Hybrid Scheduling and Cache-Aware Design:}
    Incorporating graph partitioning, memory-locality-aware relaxations, and cache-focused
    scheduling to reduce machine-level overhead.

    \item \textbf{Parallel and GPU Variants:}
    Exploring frontier-level parallelism on multi-core CPUs and massively parallel GPUs,
    especially for large-scale dense or negative-edge workloads.

    \item \textbf{Approximate or Probabilistic Extensions:}
    Introducing controlled approximation for extremely large graphs where exact distances are
    not strictly required, potentially enabling substantial additional reductions in work.
\end{itemize}

\vspace{1em}

\subsection*{Future Industrial Applications}

Although JFR is primarily motivated by theoretical and algorithmic concerns, its
combination of robustness and reduced operational footprint suggests several potential
application domains:

\begin{itemize}
    \item \textbf{Large-Scale Network Routing:}
    Efficient shortest-path updates in dense telecom and data-center networks.

    \item \textbf{Financial and Risk Analysis:}
    Handling negative-edge or irregular transaction graphs with predictable performance.

    \item \textbf{Logistics and Transportation:}
    Accelerated routing in dense transportation networks and dynamic scheduling systems.

    \item \textbf{Embedded or Resource-Constrained Systems:}
    Systems where reduced computational work directly improves longevity or responsiveness.

    \item \textbf{Dynamic or Real-Time Environments:}
    Rapid recalculation of distances under frequently changing weights, such as
    traffic navigation or adaptive grid systems.
\end{itemize}

% =================================
% APPENDICES
% =================================
\appendix
\section*{Appendix A: Amortized Proof of Theorems~\ref{thm:amortized_edges} and \ref{thm:amortized_running_time}}
\label{append:amortized_proof}

This appendix provides the rigorous derivation of the amortized bounds stated in the main text, specifically relying on the \emph{$k$-Bounded Local Multi-Hop Property} (Definition~\ref{def:k-bounded_main}).

Recall the notation: $n=|V|$, $m=|E|$. Let $s(v)$ denote the total number of times vertex $v$ is added to the active frontier $\mathcal{F}$ (activations). Let $D_v$ be the number of strict distance improvements for vertex $v$ throughout the execution ($D_v \le n-1$).

\subsection*{A.1 Edge-Inspection Decomposition}

\begin{lemma}
\label{lem:edge_decomp}
The total number of inspected frontier edges is exactly the sum of the out-degrees of activated vertices:
\[
\sum_{t\ge 1} |E_F^{(t)}| = \sum_{v\in V} s(v) \deg(v).
\]
\end{lemma}
\begin{proof}
In each iteration $t$, if $v \in F^{(t)}$, all edges $(v, u) \in E_{\text{out}}(v)$ are inspected. Summing over all $t$ is equivalent to summing over all activation events for each vertex.
\end{proof}

\subsection*{A.2 Proof of Theorem~\ref{thm:amortized_edges} (The Role of $k$-Bounded Propagation)}

\begin{lemma}[$k \Rightarrow \tau \ge k$]
\label{lem:k-implies-tau}
If $\text{LMH}$ is $k$-bounded and a vertex $v$ is stabilized by $\text{LMH}$ in iteration $t$ (i.e., no strict improvement is available via $\le k$-length paths inside $N_k(F(t))$), then any strict distance improvement to $v$ originating from outside the $k$-hop neighborhood $N_k(F(t))$ requires at least $k$ outer iterations to propagate to $v$. Consequently, the observed stability window $\tau$ satisfies $\tau \ge k$.
\end{lemma}
\begin{proof}
Any external strict improvement must traverse at least one edge to enter the $k$-hop neighborhood $N_k(F(t))$. Since $k$-bounded $\text{LMH}$ exhaustively stabilizes all $\le k$-length paths within this region, the propagation of the external effect must proceed via outer-iteration frontier expansions, each of which advances the affected region by at most one hop. Thus, at least $k$ outer iterations are required before the external update can reach $v$.
\end{proof}

\begin{proof}[Proof of Theorem~\ref{thm:amortized_edges}]
Consider a vertex $v$. Under the $k$-Bounded Local Multi-Hop Property (Definition~\ref{def:k-bounded_main}) and the resulting stability lower bound $\tau \ge k$ (Lemma~\ref{lem:k-implies-tau}), a vertex $v$ can be strictly improved at most once every $k$ outer iterations. It is only re-inserted (re-activated) upon a new strict distance decrease.

This guarantees that for every block of $k$ outer iterations, there must be at most one strict distance improvement $D_v$ that activates $v$. Since the initial activation is separate, the total number of activation phases $s(v)$ is bounded by:
\[
s(v) \le 1 + \left\lceil \frac{D_v}{k} \right\rceil.
\]
Substituting this new bound (using the algorithm parameter $k$ instead of the observable $\tau$) into Lemma~\ref{lem:edge_decomp}:
\[
\sum_{t} |E_F^{(t)}| = \sum_{v\in V} s(v) \deg(v) \le \sum_{v\in V} \left(1 + \frac{D_v}{k}\right) \deg(v).
\]
Expanding and applying the definitions $m = \sum_v \deg(v)$ and $n = |V|$:
\[
\sum_{t} |E_F^{(t)}| \le \sum_{v\in V} \deg(v) + \frac{1}{k} \sum_{v\in V} D_v \deg(v).
\]
This total inspection count is bounded by $O\left(n + m \cdot \frac{\bar{D}}{k} \right)$.
\end{proof}

\subsection*{A.3 Proof of Theorem~\ref{thm:amortized_running_time}}

\begin{proof}
The total running time $T_{\text{total}}$ is the sum of the edge relaxation cost and the total overhead of $k$-Bounded Local Multi-Hop ($\text{LMH}$) Propagation steps across all outer iterations $T$.
Let $C_{\mathrm{LMH}}(t)$ denote the computational cost of the $\text{LMH}$ propagation step in iteration $t$, which is bounded by Assumption~\ref{asm:jp_cost_revised}.

\[
T_{\text{total}}
= \sum_{t=1}^{T}
\left(
    c_{\text{relax}} \cdot |E_F^{(t)}|
    + C_{\mathrm{LMH}}(t)
\right).
\]

Substituting the edge inspection bound from Theorem~\ref{thm:amortized_edges} for the edge relaxation term:
\[
\sum_{t} c_{\text{relax}} \cdot |E_F^{(t)}| \le O\left( n + m \cdot \frac{\bar{D}}{k} \right).
\]
The total time then satisfies:
\[
T_{\mathrm{total}}
= O\!\left(
    n
    + m \cdot \frac{\bar{D}}{k}
    + \sum_{t} C_{\mathrm{LMH}}(t)
\right).
\]
This revised bound explicitly shows that the overall speedup depends on two competing factors: the savings factor $1/k$ applied to the classical relaxation term, and the accumulated overhead cost $\sum_{t} C_{\mathrm{LMH}}(t)$, which is proportional to $k$ (Assumption~\ref{asm:jp_cost_revised}). For a net speedup, the benefit from the $1/k$ reduction must outweigh the cost term.
\end{proof}
\section*{Appendix B: General Framework of JFR Strategy}
\label{sec:appendix_alg}

This appendix presents the general algorithmic framework for the JFR (Jump Frontier Relaxation with Frontier Filtering) strategy. The pseudocode details the core logical flow and component interactions. Implementation-specific factors, such as the exact ordering rule for the priority structure $\mathcal{Q}$, the topological strategy for Jump Propagation, and the precise Frontier Filtering threshold, are abstracted to maintain the framework's theoretical generality and focus on the fundamental algorithmic contribution.

\begin{algorithm}[htb!]
\caption{General Framework of JFR Strategy}
\label{alg:jfr_abstract}
\begin{algorithmic}[1]

\Require Graph $G=(V,E,w)$, source $s$, depth parameter $k$
\Ensure Distance vector $\mathbf{d}$, Parent pointers $\pi$

\State \textbf{Initialize:} $\mathbf{d}[v]\gets +\infty,\ \pi[v]\gets\text{NIL}$ for all $v\in V$; $\mathbf{d}[s]\gets 0$
\State \textbf{Initialize:} Frontier $\mathcal{F}\gets\{s\}$, Priority Structure $\mathcal{Q}\gets\{s\}$
\State \textbf{Initialize:} Auxiliary metadata $\mathbf{aux}$ \Comment{Tracks update history and local stability}

\While{$\mathcal{F}\neq\emptyset$}
  \State $u \gets \mathcal{Q}.\text{select next active}()$ \Comment{Selection based on priority metric}

  \If{\Call{MeetsStabilityCriterion}{$u$, $\mathbf{aux}$}}
    \Comment{Checks if node is locally stable (implies $\tau \ge k$, see Lemma~\ref{lem:k-implies-tau})}
    \State \Call{Local Multi-Hop Propagation}{$u$, $\mathbf{d}$, $\pi$, $\mathcal{F}, k$}
    \Comment{Performs $k$-bounded updates per Definition~\ref{def:k-bounded_main}}
  \EndIf

  \ForAll{$(u,v)\in E_{\text{out}}(u)$}
    \State $d_{\text{new}} \gets \mathbf{d}[u] + w(u,v)$
    \If{$d_{\text{new}} < \mathbf{d}[v]$}
      \State $\mathbf{d}[v]\gets d_{\text{new}}$; $\pi[v]\gets u$
      \State \Call{UpdateMetadata}{$v$, $\mathbf{aux}$}
      \If{$v\notin\mathcal{F}$}
        \State $\mathcal{F}.\text{insert}(v)$; $\mathcal{Q}.\text{insert}(v)$
      \Else
        \State $\mathcal{Q}.\text{decrease\_key}(v)$
      \EndIf
    \EndIf
  \EndFor

  \If{\Call{EvaluateFilteringCondition}{$\mathcal{F}$, $\mathbf{aux}$}}
    \Comment{Adaptive criterion based on frontier density and $k$}
    \State \Call{FilterStableVertices}{$\mathcal{F}$, $\mathcal{Q}$}
    \Comment{Prunes redundant nodes leveraging $k$-stability}
  \EndIf

\EndWhile

\State \Return $\mathbf{d}, \pi$

\end{algorithmic}
\end{algorithm}

\subsection{Mechanism Rationale: Bounded Local Propagation}
\label{appendix:mechanism}

To theoretically justify \textbf{Assumption~\ref{asm:jp_cost_revised}} ($k$-LMH Cost Bound) and the mechanisms described in \textbf{Definition~\ref{def:k-bounded_main}} ($k$-Bounded Local Multi-Hop Propagation) without loss of generality, we describe the logical control flow of the \textit{Local Multi-Hop Propagation} mechanism utilized in our implementation.

\paragraph{Bounded Depth Constraint}
The \textit{Local Multi-Hop Propagation} procedure is a \textbf{depth-limited local relaxation process} applied on the induced subgraph $G[N_k(F^{(t)})]$, where $N_k(F^{(t)})$ denotes the $k$-hop neighborhood of the active frontier. By limiting the propagation depth to a small constant $k$ (or a heuristic bound derived from structural indicators), the computational cost at iteration $t$ is tightly controlled by the local topology around the frontier.
Consistent with Assumption~\ref{asm:jp_cost_revised}, the cost is bounded by the size of the local neighborhood and the depth parameter $k$:
\begin{equation}
C_{\mathrm{LMH}}(t) = \mathcal{O}\Bigg( k \cdot \sum_{v \in N_k(F^{(t)})} \mathrm{deg}(v) \Bigg).
\end{equation}
This bounded-depth design ensures that each iteration remains efficient. The cost remains proportional to the \textbf{local frontier neighborhood size} rather than the global graph size $|E|$, satisfying the conditions for the cost-benefit tradeoff derived in Theorem~\ref{thm:amortized_running_time}.

\paragraph{Ensuring Stability}
The mechanism achieves the stability guarantees proved in \textbf{Lemma~\ref{lem:k-implies-tau}} ($\tau \ge k$) by enforcing \textit{local convergence} within the depth-limited region before releasing vertices. Specifically, a vertex $v$ is only removed from the frontier (contracted) after the local relaxation process stabilizes its distance value against all paths of length $\le k$ within the window. Consequently, $d(v)$ cannot be improved again until a relaxation wave propagates from outside this local window, effectively guaranteeing stability for at least $\tau \ge k$ subsequent outer iterations.



\end{document}